\documentclass[a4paper,10pt]{article}
\usepackage[utf8]{inputenc}
\usepackage{graphicx, amsmath, amsthm, amssymb, algorithmic,url,multirow}
\usepackage{color}
\usepackage{latexsym}
\usepackage{hyperref}
\usepackage{caption}
\usepackage{subcaption}
\usepackage[format=plain,justification=raggedright,singlelinecheck=false]{caption} 
\usepackage{array}
\newcolumntype{H}{>{\setbox0=\hbox\bgroup}c<{\egroup}@{}}

\usepackage[linesnumbered,ruled,norelsize]{algorithm2e} 
\makeatletter
\newcommand{\removelatexerror}{\let\@latex@error\@gobble}
\makeatother
\usepackage{geometry}
\emergencystretch=\maxdimen
\newtheorem{theorem}{Theorem}
\newtheorem{lemma}{Lemma}
\newtheorem{formula}{Formula}

\title{No-hole $\lambda$-$L(k, k-1, \ldots, 2, 1)$-labeling for Square Grid}

\author{
        Soumen Atta\thanks{Part of this work has done by the author when he has been in the Faculty of Mathematics and Computer Science, University of \L\'od\'z, \L\'od\'z, Poland as an Erasmus+ Exchange PhD Student.}\\
        Department of Computer Science and Engineering,\\
        University of Kalyani, Kalyani, West Bengal, India\\
        Email: soumen.atta@klyuniv.ac.in
        \and
        Priya Ranjan Sinha Mahapatra\thanks{The author is partially supported by DST-PURSE scheme at University of Kalyani, India.}\\
        Department of Computer Science and Engineering,\\
        University of Kalyani, Kalyani, West Bengal, India\\
        Email: priya@klyuniv.ac.in
        \and
        Stanis{\l}aw Goldstein\\
        Faculty of Mathematics and Computer Science,\\University of \L\'od\'z, \L\'od\'z, Poland\\
        Email: goldstei@math.uni.lodz.pl
}

\date{}

\begin{document}

\maketitle

\begin{abstract}
Given a fixed $k$ $\in$ $\mathbb{Z}^+$ and $\lambda$ $\in$ $\mathbb{Z}^+$, the objective of a $\lambda$-$L(k, k-1, \ldots, 2, 1)$-labeling of a graph $G$ is to assign non-negative integers (known as labels) from the set $\{0, \ldots, \lambda-1\}$ to the vertices of $G$ such that the adjacent vertices receive values which differ by at least $k$, vertices connected by a path of length two receive values which differ by at least $k-1$, and so on. The vertices which are at least $k+1$ distance apart can receive the same label.
The smallest $\lambda$ for which there exists a $\lambda$-$L(k, k-1, \ldots, 2, 1)$-labeling of $G$ is known as the $L(k, k-1, \ldots, 2, 1)$-labeling number of $G$ and is denoted by $\lambda_k(G)$. The ratio between the upper bound and the lower bound of a $\lambda$-$L(k, k-1, \ldots, 2, 1)$-labeling is known as the approximation ratio.
In this paper a lower bound on the value of the labeling number for square grid is computed and a formula is proposed which yields a $\lambda$-$L(k, k-1, \ldots, 2, 1)$-labeling of square grid, with  approximation ratio at most $\frac{9}{8}$. The labeling presented is a no-hole one, i.e., it uses each label from $0$ to $\lambda-1$ at least once.
\end{abstract}





\section{Introduction} \label{intro}
The {\em frequency assignment problem} ({\em FAP}) is a problem of assigning frequencies to different radio transmitters so that no interference occurs~\cite{hale1980frequency}. This problem is also known as the {\em channel assignment problem} (CAP)~\cite{dubhashi2002channel, bertossi2003channel}. Frequencies are assigned to different radio transmitters in such a way that comparatively close transmitters receive frequencies with more gap than the transmitters which are significantly apart from each other.
Motivated by this problem of assigning frequencies to different transmitters, Yeah~\cite{yeh1990} and after that Griggs and Yeh~\cite{griggs1992labelling} proposed an $L(2,1)$-labeling for a simple graph. An $L(2,1)$-labeling of a graph $G$ is a mapping $f : V(G) \rightarrow \mathbb{Z}^+$ such that $|f(u) - f(v)| \geq 2$ when $d(u,v) = 1$, and $|f(u) - f(v)| \geq 1$ when $d(u,v) = 2$, where $d(u,v)$ denotes the minimum path distance between the two vertices $u,~v~\in~V$ (One can use the same label if the distance between two vertices is greater than $2$) \cite{griggs1992labelling, chang19962, bodlaender2000lambda, chang2000d, shao20072}.

Various generalizations of the original problem, for diverse types of graphs, finite or infinite, has been described in the literature \cite{chang2003distance, georges2003labeling, calamoneri2004h, fishburn2003no, yeh2006survey, kuo20042, calamoneri2009h, calamoneri2006optimal, georges1994relating, shao20082}.
Instead of $L(2,1)$-labeling one can consider $L(3,2,1)$-labeling, and more generally an $L(k,k-1, \dots, 1)$-labeling. Nandi et al.~\cite{nandi2015k} considered an $L(k,k-1, \dots, 1)$-labeling for a triangular lattice.

In this paper $L(k,k-1, \dots, 2, 1)$-labeling for a square grid is considered. The definition of the problem is given in Section~\ref{def}. The lower bound on the value of $\lambda_k$, the labeling number for the square grid, is derived in Section~\ref{lowerbound}. In Section~\ref{formula}, a formula is given that attaches a label to any vertex of an infinite square grid for arbitrary values of $k$. The correctness proof of the proposed formula is given Section~\ref{proof}. In Section~\ref{no-hole} we prove that the proposed formula gives a no-hole labeling. Our $\lambda$-labeling yields immediately an upper bound on $\lambda_k$, given together with the approximation ratio implied by the proposed formula in Section~\ref{upper}. Finally, the paper is concluded in Section~\ref{conclusion}.

\section{Problem Definition} \label{def}
Let $G=(V,E)$ be a graph with a set of vertices $V$ and a set of edges $E$, and let $d(u,v)$ denote the shortest distance between vertices $u,~v~\in~V$. Given a fixed $k$ $\in$ $\mathbb{Z}^+$ and $\lambda$ $\in$ $\mathbb{Z}^+$, a $\lambda$-$L(k, k-1, \ldots, 2, 1)$-labeling of the graph is a mapping $f: V \rightarrow \{0, \dots, \lambda-1\}$ such that the following inequalities are satisfied:
\[\lvert f(x)-f(y) \rvert \geq \left\{
  \begin{array}{lr}
     k & : d(x,y) = 1\\
    k-1 & : d(x,y) = 2\\
    \vdots \\
    1 & : d(x,y) = k,\\
  \end{array}
\right.
\]
which can be written more compactly as

\begin{equation}\tag{*}
\lvert f(x)-f(y) \rvert \geq k+1-d(x,y) \text{ for }x \neq y.
\end{equation}

We shall call any function $f:V\to \mathbb{Z}$ satisfying the inequality a \emph{labeling function}.

If the distance between two vertices is at least $k+1$, the same label can be used for both of them. This minimum distance is known as the {\em reuse distance}~\cite{nandi2015k}. The $L(k, k-1, \ldots, 2, 1)$-labeling number for the graph, denoted by $\lambda_k$, is the minimum $\lambda$ for which a valid $\lambda$-$L(k, k-1, \ldots, 2, 1)$-labeling for the graph exits. Hence, our objective is to find, for each $k$, a no-hole $\lambda$-$L(k, k-1, \ldots, 2, 1)$-labeling with $\lambda$ as close to $\lambda_k$ as possible.

We consider an infinite planar square grid $G=(V,E)$ with the set of vertices $V=\mathbb{Z}\times\mathbb{Z}$ and the set of edges $E=\{\{u,v\}\colon u=(u_1,u_2), v= (v_1,v_2), \text{and either } |u_1-v_1|=1, u_2=v_2 \text{ or } u_1=v_1,  |u_2-v_2|=1\}.$ It will be called \emph{`the square grid'} in the sequel. The distance between $u$ and $v$ used in the sequel is the \emph{Manhattan distance}: $d(u,v)=|u_1-v_1|+|u_2-v_2|.$

\section{Lower Bound on $\lambda_k$} \label{lowerbound}
\begin{theorem} \label{theorem-LB}
For $k\geq1$,
\begin{equation*}
  \lambda_k \geq
  \begin{cases}
    \frac{2}{3}p(p+1)(2p+1)+2 &\text{if $k=2p$ is even,}\\
    \frac{2}{3}p(p+1)(2p+3)+2 &\text{if $k=2p+1$ is odd.}
  \end{cases}
\end{equation*}
\end{theorem}

\begin{proof}
We start with the case of even $k=2p$. We shall write $B_m$ for the ball $\{u\in V\colon d(0,u)\leq m\}$, and $S_m$ for the sphere $\{ u\in V\colon d(0,u)=m\}$ (here $0=(0,0)$). Note that there is just one point in $S_0$ and $4m$ points in $S_m$ for $m>0$ (See Fig.~\ref{evenclique}). It is easy to calculate that there are exactly $1+4+\ldots+4m= 2m^2+2m+1$ points in $B_m$. To obtain a lower bound on the $L(k, k-1, \ldots, 2, 1)$-labeling number, we identify the smallest interval containing all integers needed to label the vertices in the ball $B_p$. To this aim, we use a labeling function $f: V\to \mathbb{Z}$. It is clear that $\lambda_k\geq \max f(B_p)-\min f(B_p)+1$.

\begin{figure}[h]
 \centering
 \includegraphics{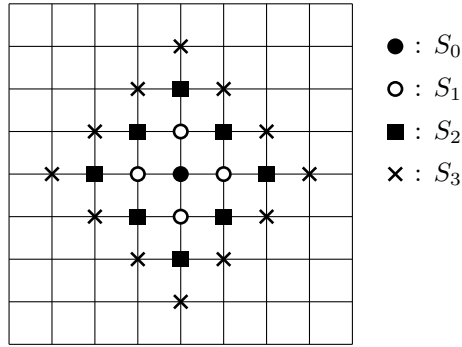}
 \caption{$S_m$ when $m = 0, 1, 2, 3$.}
 \label{evenclique}
\end{figure}

Let us put all the values of the function $f$ on $B_p$ in increasing order:
$z_0<z_1<\ldots< z_n$. We have $\lambda_k\geq z_n-z_0+1$. Note that because of (*), the function $f$ is injective on $B_p$, hence $n=2p^2+2p$ is one less than the number of points in $B_p$. Let $u_i=f^{-1}(z_i)$ and and let $q,r$ be such that $u_0\in S_q, u_n\in S_r$.

The method of obtaining the lower bound is a formalization of that used by Nandi et al.~\cite{nandi2015k}. According to (*), $z_{i+1}-z_i\geq 2p+1-\max \{d(u_i, v)\colon v\in B_p\setminus\{u_i\}\}$. If $u_i\in S_m$, then $\max \{d(u_i, v)\colon v\in B_p\setminus\{u_i\}\}=m+p$, hence
$z_{i+1}-z_i\geq p+1-m$. Considering $z_i$ for $i=0,1,\ldots,n-1$, we can already estimate that
\[
z_n-z_0= (z_1-z_0)+\ldots +(z_n-z_{n-1})\geq |S_p|+2|S_{p-1}|+\ldots
+p|S_1|+(p+1)|S_0|-(p+1-r).
\]

Let us call the number on the RHS of the inequality $c_p$. Now, if a point $u_i$ is such that $i<n$ and $u_{i+1}\in B_{p-1}$, then $z_{i+1}-z_i\geq 2p+1 - \max \{d(u_i, v)\colon v\in B_{p-1}\setminus\{u_i\}\}=p+2-m$ (instead of $p+1-m$). There are at least $|B_{p-1}|$ points like this if $q =p$, and $|B_{p-1}|-1$ if $q\neq p$, and the RHS of the inequality above can be increased by the amount. Continuing further in this manner, we get
\begin{align*}
z_n-z_0 & \geq c_p+ (|B_{p-1}|-1)+\ldots+(|B_{q}|-1)+|B_{q-1}|+\ldots+|B_0| \\
   & =c_p+|S_{p-1}|+2|S_{p-2}|+\ldots+(p-1)|S_1|+p|S_0|-(p-q) \\
   & = 4\big(\sum_{m=1}^{p}m(p+1-m)+\sum_{m=1}^{p-1}m(p-m)\big)+(r+q).
\end{align*}
Using
\[
1\cdot p+2\cdot (p-1)+\ldots+(p-1)\cdot 2 +p\cdot 1=p(p+1)(p+2)/6,
\]
and the fact that $r+q$ is at least $1$, which happens if $p,q\in\{0,1\}$ (note that they must be different, since there is only one point in $S_0$),
we easily get $\lambda_k\geq \frac{2}{3}p(p+1)(2p+1)+2.$

Now, if $k=2p+1$ is odd, each of the $2p^2+2p$ summands $z_1-z_0, z_2-z_1, \ldots, z_n-z_{n-1}$ is larger by one, hence $\lambda_k\geq \frac{2}{3}p(p+1)(2p+3)+2.$ A better estimate can be obtained by considering the set $T_0=\{(0,0), (0,1)\}$ and, for $m>0$, the sets $T_m=\{u\in \mathbb{Z}\times\mathbb{Z}\colon d(u,T_0)=m\}$ (see Fig.~\ref{oddclique}). This, however, does not change the asymptotic behavior of $\lambda_k$.
\begin{figure}[h]
 \centering
 \includegraphics{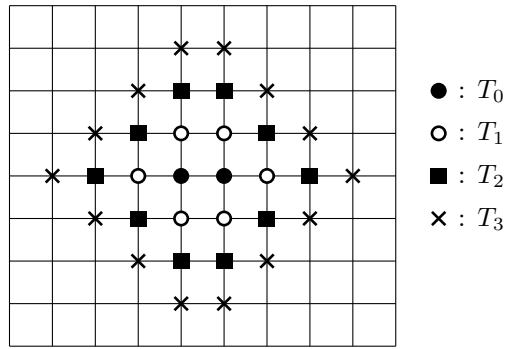}
 \caption{$T_m$ when $m = 0, 1, 2, 3$.}
 \label{oddclique}
\end{figure}

\end{proof}

\section{Proposed Formula} \label{formula}
In this section a formula is given to find the label of any vertex of the square grid under $L(k, k-1, \ldots, 2, 1)$-labeling for general $k$. Let the label assigned to the vertex $v(x, y)$ is denoted by $L(x, y)$. Formula~\ref{formulae} gives the definition of $L(x, y)$.

\begin{formula} \label{formulae}
 \begin{equation*}
 \small
  L(x,y) =
  \begin{cases}
    [(2p+3)x + (3p^2+7p+5)y] \mod \frac{1}{2}(p+1)(3p^2+5p+4) &\text{if $k=2p+1$ and $p(\geq 1)$ is odd;}\\
    [(2p+3)x + (3p^2+6p+3)y] \mod \frac{1}{2}(3p^3+8p^2+8p+4) &\text{if $k=2p+1$ and $p(\geq 0)$ is even;}\\
    [(2p+1)x + (3p^2+4p+2)y] \mod \frac{1}{2}(3p^3+5p^2+5p+1) &\text{if $k=2p$ and $p(\geq 3)$ is odd;}\\
    [(2p+1)x + (3p^2+3p+1)y] \mod \frac{1}{2}p(3p^2+5p+4) &\text{if $k=2p$ and $p(\geq 2)$ is even.}\\
  \end{cases}
\end{equation*}
\end{formula}

Note that many correct labelings may exist when the coefficients of $x$ and $y$ are restricted to be co-prime. If this restriction is removed then correct labelings also exist with reduced $\lambda_k$. Thus we have considered all possible combinations of the coefficients for $x$ and $y$ at the time of designing Formula~\ref{formulae} for finding a labeling with the minimum $\lambda_k$. The assignment of labeling for $k=7$ is shown in Fig.~\ref{LabelingExample} for some vertices.

\begin{figure}[htbp]
 \centering
 \includegraphics{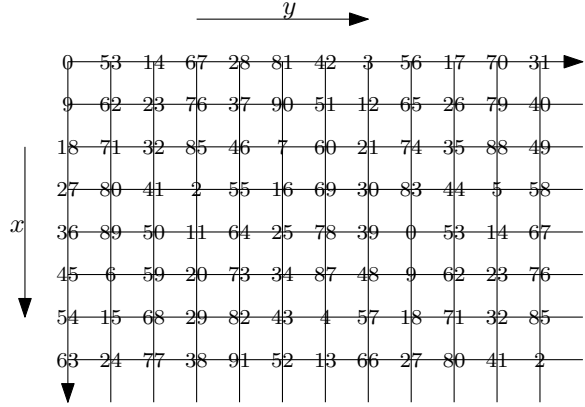}
 \caption{Assignment of labeling for $k=7$}
 \label{LabelingExample}
\end{figure}

\subsection{Correctness Proof of the Proposed Formula} \label{proof}
Formula~\ref{formulae} is said to be correct if and only if the inequality constraints of the problem mentioned in Section~\ref{def} are satisfied. The proof of Theorem~\ref{theorem-2} shows the correctness of Formula~\ref{formulae}.
Lemma~\ref{lemma-3} is needed to prove Theorem~\ref{theorem-2}.

\begin{theorem} \label{theorem-2}
Formula ~\ref{formulae} yields a $\lambda$-$L(k, k-1, \ldots, 2, 1)$-labeling of the square grid, with
\begin{equation}
  \lambda =
  \begin{cases}
    \frac{1}{2}(p+1)(3p^2+5p+4) &\text{if $k=2p+1$ and $p(\geq 1)$ is odd;}\\
    \frac{1}{2}(3p^3+8p^2+8p+4) &\text{if $k=2p+1$ and $p(\geq 0)$ is even;}\\
    \frac{1}{2}(3p^3+5p^2+5p+1) &\text{if $k=2p$ and $p(\geq 3)$ is odd;}\\
    \frac{1}{2}p(3p^2+5p+4) &\text{if $k=2p$ and $p(\geq 2)$ is even.}\\
  \end{cases}\tag{**}
\end{equation}

More precisely, if $|x_1-x_2|+|y_1-y_2| = r$, then $|L(x_1,y_1)-L(x_2,y_2)| \geq k+1-r$, where $0 < r \leq k+1$ and $L(x,y)$ is defined by Formula~\ref{formulae}.
\end{theorem}

\begin{lemma} \label{lemma-3}
Let $a$, $b$, $c$ $\in$ $\mathbb{Z}^+$ and $L(x,y)=(ax+by) \mod c$. Now for any $x_1$, $y_1$, $x_2$, $y_2$ $\in$ $\mathbb{Z}$, if $L(x_1,y_1)>L(x_2,y_2)$ then $|L(x_1,y_1)-L(x_2,y_2)|=L(x_1-x_2,y_1-y_2)$.
\end{lemma}

\begin{proof}
Clearly $0\leq L(x,y)<c$ for any $x$, $y$ $\in$ $\mathbb{Z}$. Hence, $0 \leq |L(x_1,y_1)-L(x_2,y_2)| < c$. Again, for any $A$, $B$ $\in$ $\mathbb{Z}$, ($A \mod c - B \mod c$) mod $c = $($A-B$) $\mod$ $c$.
Put $A=ax_1+by_1$ and $B=ax_2+by_2$. Then $|L(x_1,y_1)-L(x_2,y_2)| = A \mod c - B$ mod $c$ $=$ ($A \mod c - B \mod c$) mod $c $=$ $($A-B$) $\mod$ $c$ $=$ $L(x_1-x_2,y_1-y_2)$.
\end{proof}

\begin{proof}[Proof of Theorem~\ref{theorem-2}]
We prove it for $L(x, y) = [(2p+3)x + \left(3p^2+7p+5\right)y] \mod \frac{1}{2}(p+1)(3p^2+5p+4)$ and $k=2p+1$, $p (\geq 3)$ is odd, and show the correctness for $p=1$ separately. The correctness of Formula~\ref{formulae} can be proved for other values of $k$ in a similar way.

We can change the order of ($x_1,y_1$) and ($x_2,y_2$) in such a way that $L(x_1,y_1) \geq L(x_2,y_2)$, since exchanging indices $1$  and $2$ does not change $r$. By Lemma~\ref{lemma-3} we have to show that for $x$, $y$ $\in \mathbb{Z}$ with $|x|+|y|=r$, $L(x,y) \geq k+1-r$. Note that the inequality is always satisfied for $r=k+1$. Hence, we can assume $0 < r < k+1$.

Put $a=2p+3, b = 3p^2+7p+5$ and $c=\frac{p+1}{2}(3p^2+5p+4)$. Note that $|ax+by| < 5c$ for any $x$, $y$ with $|x|+|y|=r$.

{\bf Case-I} Assume that $ct \leq by \leq ax+by < c(t+1)$ for some $t \in [-5,4] \cap \mathbb{Z}$.  Then
\[
(ax+by) \mod c = ax+by-ct \geq ax > 2p+2.
\]
(Since $x > 0$, $ax \geq a = 2p+3$.) Hence, $L(x,y) > 2p+2 = k+1 \geq k+1-r$.

{\bf Case-II} Assume that $x=0$. Let $Y_t= \{y : ct \leq by < c(t+1)\}$ and $y_t = \min (Y_t), t \in[-5,4] \cap \mathbb{Z}$ for $|y_t| \leq k$. Note that $b>0$, so that whenever $L(x,y_t) \geq k+1$, also $\forall y\in Y_t$, $L(x,y) \geq k+1$.
Since $y\neq0$ (we already have $x=0$), we have $y_0=1$ and $by_0 \mod c=b>2p+2=k+1$. Hence, we need only consider $t\neq0$.
Put $d=\frac{2p^2+3p+1}{6p^2+14p+10}=\frac{p+1}{2}\frac{2p+1}{b}$. Note that for each odd $p\neq1$, $\frac{1}{4}<d<\frac{1}{3}.$
Now $y_t\geq ct/b = t(\frac{p+1}{2}-d)$, so that $y_t=t\frac{p+1}{2}+e$, where
\begin{equation*}
  e=
  \begin{cases}
    0 &\text{if $t=1,2$ or $3$;}\\
    -1 &\text{if $t=4$;}\\
    1 &\text{if $t=-1,-2$ or $-3$;}\\
    2 &\text{if $t=-4$ or $t=-5$.}
  \end{cases}
\end{equation*}

We have $L(0,y_t)=by_t-ct=t(b\frac{p+1}{2}-c)+be=t(2p^2+3p+1)/2 + be$. The inequality $L(0,y_t)\geq 2p+2$ is obviously true if $t$ is positive and $e=0$.
If $t=4$, we have $L(0,y_t)=2(2p^2+3p+1)-b=p^2-p-3 \geq 2p+2$ for odd $p\geq5$, and $L(0,y_t) \geq k+1-r$ for $p=3$. For $t=-1, -2$ or $-3$, it is enough to check the ``worst'' case,
namely $t=-3$, which yields $L(0,y_t)=(5p+7)/2 \geq 2p+2$. Again, we can omit $t=-4$ and check that for $t=-5$ we get $L(0,y_t)=(2p^2+13p+15)/2 \geq 2p+2.$


{\bf Case-III} Assume that $by < ct \leq ax+by$.  Note that then $c(t-1)<by<ct \leq ax+by <c(t+1)$. We will show that there exist at most two $y$'s satisfying the inequality. Let
$y_t =$max$\{y: by<ct~\wedge~(\exists x: ct\leq ax+by)\}$.
Thus $by_t<ct\leq ax+by_t$ for some $x$.
Suppose $b(y_t - 2) < ct \leq ax+b(y_t - 2)$ for some $x$.
Then $ax+b(y_t-1)=(ax-2b)+by_t \geq ct>by_t$.
But $ax-2b \leq a(2p+2)-2b=2[(p+1)(2p+3)-(3p^2+7p+5)] = 2(-p^2-2p-2)<0$, which is a contradiction.
If we find $x_t=\min\{x:by_t<ct\leq ax+by_t\}$ and $x'_t=\min\{x:b(y_t-1)<ct\leq ax+b(y_t-1)\}$ and if $|x_t|+|y_t|<2p+2$ (similarly $|x'_t|+|y_t|<2p+2$), then it is enough to check that $L(x_t,y_t) \geq k+1-r$ and $L(x'_t,y_t-1)\geq k+1-r$.

Put $d=\frac{2p^2+3p+1}{6p^2+14p+10}=\frac{p+1}{2}\frac{2p+1}{b}$. Note that for each odd $p\neq1$, $\frac{1}{4}<d<\frac{1}{3}.$
Now $y_t < ct/b = t(\frac{p+1}{2}-d)$, so that $y_t=t\frac{p+1}{2}+e$, where
\begin{equation*}
  e=
  \begin{cases}
    -1 &\text{if $t=1, 2$ or $3$;}\\
    -2 &\text{if $t=4$;}\\
     0 &\text{if $t=-1, -2, -3$ or $-4$;}\\
     1 &\text{if $t=-5$.}
  \end{cases}
\end{equation*}

Using $ct \leq ax_t + by_t \Rightarrow x_t \geq \frac{ct - by_t}{a}$, and $L(x_t, y_t) = ax_t + by_t -ct$,  we construct Tab.~\ref{case-3.1}. Whenever $|y_t|$, $|x_t|$ or $r$ is at least $2p+2$, there is no need for further calculation, and the respective positions are filled with dashes.

\begin{table}[h]
\centering
\caption{}
\small
\label{case-3.1}
\begin{tabular}{cccccc}
\hline
\textbf{$t$} & \textbf{$y_t$}      & \textbf{$x_t$}    & \textbf{$r = |x_t| + |y_t|$}         & \textbf{$k+1-r$} & \textbf{$L(x_t, y_t)$} \\ \hline
$1$                     & $\frac{p-1}{2}$           & $(p+2)$                   & $\frac{3}{2}(p+1)$      & $\frac{1}{2}(p-3)$          & $\frac{3}{2}(p+1)$       \\
$2$                     & $p$                       & $\frac{(p+3)}{2}$         & $\frac{3}{2}(p+1)$      & $\frac{1}{2}(p-3)$          & $\frac{1}{2}(p+1)$       \\
$3$                     & $\frac{(3p+1)}{2}$        & $2$                       & $\frac{1}{2}(3p+5)$     & $\frac{1}{2}(p-1)$          & $\frac{1}{2}(3p+5)$      \\
$4$                     & $2p$                      & $(p+1)$                   & $3p+1$                  & $-$                         & $-$                      \\
$-1$                    & $-\frac{(p+1)}{2}$        & $\frac{(p+1)}{2}$         & $p+1$                   & $p+1$                       & $p+1$                    \\
$-2$                    & $-(p+1)$                  & $(p+1)$                   & $2(p+1)$                & $-$                         & $-$                      \\
$-3$                    & $-\frac{3(p+1)}{2}$       & $(3p+1)$                  & $\frac{1}{2}(9p+5)$     & $-$                         & $-$                      \\
$-4$                    & $-2(p+1)$                 & $-$                       & $-$                     & $-$                         & $-$                      \\
$-5$                    & $-\frac{(5p+3)}{2}$       & $-$                       & $-$                     & $-$                         & $-$                      \\ \hline
\end{tabular}
\end{table}

Using $ct \leq ax'_t + b(y_t - 1) \Rightarrow x'_t \geq \frac{ct - b(y_t - 1)}{a}$, and $L(x'_t, y_t-1) = ax'_t + b(y_t - 1) -ct$, we construct Tab.~\ref{case-3.2} with the corresponding values. As above, we use dahses whenever $|y_t - 1|$, $|x'_t|$ or $r$ is at least $2p+2$, and there is no need for further calculation.

\begin{table}[!h]
\centering
\caption{}
\small
\label{case-3.2}
\begin{tabular}{cccccc}
\hline
\textbf{$t$} & \textbf{$y_t - 1$}      & \textbf{$x'_t$}    & \textbf{$r = |x'_t| + |y_t - 1|$}         & \textbf{$k+1-r$} & \textbf{$L(x'_t, y_t - 1)$} \\ \hline
$1$                     & $\frac{p-3}{2}$           & $\frac{(5p+7)}{2}$        & $3p+2$                  & $-$                         & $-$                      \\
$2$                     & $(p-1)$                   & $2p+3$                    & $-$                     & $-$                         & $-$                      \\
$3$                     & $\frac{(3p-1)}{2}$        & $\left\{
	\begin{array}{ll}
		\frac{3(p+1)}{2},  & \mbox{if } p = 3,5 \\
		\frac{3p+1}{2},  & \mbox{if } p (\geq 7)
	\end{array}
\right.$                      & $-$                     & $-$                         & $-$     \\
$4$                     & $2p-1$                    & $\frac{(5p+9)}{2}$        & $-$                     & $-$                         & $-$                      \\
$-1$                    & $-\frac{(p+3)}{2}$        & $2p+2$                    & $-$                     & $-$                         & $-$                   \\
$-2$                    & $-(p+2)$                  & $\frac{(5p+3)}{2}$        & $-$                     & $-$                         & $-$                      \\
$-3$                    & $-\frac{(3p+5)}{2}$      & $(3p+2)$                  & $-$                     & $-$                         & $-$                      \\
$-4$                    & $-(2p+3)$                 & $-$                       & $-$                     & $-$                         & $-$                      \\
$-5$                    & $-\frac{5(p+1)}{2}$       & $-$                       & $-$                     & $-$                         & $-$                      \\ \hline
\end{tabular}
\end{table}

{\bf Case-IV} Assume that $ax+by < ct \leq by$, where $t$ $\in$ $[-4,4] \cap \mathbb{Z}$.
Then $c(t-1) < ax+by <ct \leq by < c(t-1)$ and $ax+by \geq ax+ct = c(t-1) +(ax+c)$.
Hence, $L(x,y)$ $= (ax+by) \mod c$ $= ax+c$.

Since $ax \geq a(-2p-2) = -2(2p+3)(p+1)$, we have

$L(x,y) = \frac{\left(p+1\right)\left(3p^2+5p+4\right)}{2} - 2(2p+3)(p+1) = \frac{3}{2}p^3-\frac{11}{2}p-4 \geq 2p+2$, for $p \geq 3$.

Therefore, for $p \geq 3$, $L(x,y) \geq k+1-r$.

{\bf Case-V} Assume that $x < 0$, $ax+by \geq ct$ and $by < c(t+1)$.

Let $Y_t = \{y : \exists~x \text{ s.t. } ct \leq ax+by < by < c(t+1)\}$.
Then it is enough to check the inequality for $y_t = \min (Y_t)$ and for $y_t + 1$, and for them we should check if for
$x_t = \min \{x : ct \leq ax+by_t < by_t < c(t+1)\}$ and
$x'_t = \min \{x : ct \leq ax+b(y_t+1) < b(y_t+1) < c(t+1)\}$.

Thus we need to check $L(x_t,y_t) \geq k+1-r$ and $L(x'_t,y_t+1) \geq k+1-r$.

Using $by_t < c(t+1)$, we construct Tab.~\ref{case-5}.

\begin{table}[h]
\centering
\caption{}
\small
\label{case-5}
\begin{tabular}{c|ccccccccc}
\hline
\textbf{$t$}   & $1$ & $2$ & $3$ & $4$ & $-1$ & $-2$ & $-3$ & $-4$ & $-5$ \\ \hline
\textbf{$y_t$} & $p$ & $\frac{(3p+1)}{2}$ & $2p+1$ & $\frac{(5p-1)}{2}$ & $-1$  & $-\frac{(p+1)}{2}$  & $-(p+1)$  & $-\frac{3(p+1)}{2}$  & $-(2p+1)$  \\ \hline
\end{tabular}
\end{table}

If we calculate the values of $x_t$ and $x'_t$ from $ct \leq ax_t+by_t$ and $ct \leq ax'_t+b(y_t+1)$ respectively, then $x_t$ and $x'_t$ are always greater than $2p+2$. This completes the proof for $p \geq 3$.

{\bf Case $p=1.$} Then $k=3$ and $L(x,y)=(5x+15y)$ mod $12$.
We just need to consider different values of $x$ and $y$ such that $x$ $\in$ $\{-3,-2,-1\}$ and $y$ $\in$ $\{-3,-2,-1,0,1,2,3\}$. Clearly when ($x,y$) $\in \{(-3,-3),(-3,-2),(-3,-1),(-3,1),$ $(-3,2),(-3,3),(-2,-3),$ $(-2,-2),$ $(-2,3),(-2,2),$ $(-1,3),$ $(-1,-3)\}$, we don't need to check anything because $r = |x|+|y| \geq 4$. When $(x,y) =(-3,0), L(x,y)= 9$ and $k+1-r=1$. Similarly, when ($x,y$) $\in \{(-2,-1),(-2,0),(-2,1),$ $(-1,-2),(-1,-1),(-1,0),$ $(-1,1),$ $(-1,2)\}$, $L(x,y) \geq k+1-r$.

Hence, we always have $L(x,y) \geq k+1-r$.

\end{proof}

\subsection{No-hole Labeling Proof} \label{no-hole}
\begin{lemma} \label{lemma-4}
Formula~\ref{formulae} gives no-hole labeling.
\end{lemma}

\begin{proof}
Formula~\ref{formulae} is of the form $(ax + by) \mod c$, with $a, b$ and $c$ depending on parity of $k$ and $p$. We shall show that it is enough to check
that $\gcd(a,b,c)$ is $1$. In fact, let $m=\gcd(a,b)$ and denote by $(m)$ the principal ideal in $\mathbb{Z}$ generated by $m$. It is well known (and easy to see)
that the set $\{ax+by\colon x,y\in\mathbb{Z}\}$ equals $(m)$.  Now, if $ \gcd(m,c)  = \gcd(a,b,c) = 1$, then
$mu+cv=1$ for some $u,v\in\mathbb{Z}$. If $k\in \{0,1,\ldots,c-1\}$, then $kmu + kcv = k$, so that $kmu \equiv k \mod c$. But $kmu\in (m)$, which means
that for some $x,y\in \mathbb{Z}, (ax+by) \mod c = k$, and all integer values from $0$ up to $c-1$ are attained.

We note  the values of $\gcd(a,b)$ for different values of $k$.

 \begin{equation*}
  \gcd(a,b) =
  \begin{cases}
    1 \text{~or~} 5 &\text{if $k=2p+1$ and $p(\geq 1)$ is odd;}\\
    1 \text{~or~} 3 &\text{if $k=2p+1$ and $p(\geq 0)$ is even;}\\
    1 \text{~or~} 3 &\text{if $k=2p$ and $p(\geq 3)$ is odd;}\\
    1 &\text{if $k=2p+1$ and $p(\geq 2)$ is even.}\\
  \end{cases}
\end{equation*}

Consider the case when $k=2p+1$ and $p(\geq 1)$ is odd. In this case $a=2p+3, b=3p^2+7p+5$ and $c=\frac{1}{2}(p+1)(3p^2+5p+4)$. If $\gcd(a,b) = 1$, $\gcd(a,b,c) = 1$, and there is nothing to prove. If $\gcd(a,b) = 5$, then $p$ is congruent to $1$ modulo $5$, and $c$ is congruent to $2$ modulo $5$. So, $c$ is not divisible by $5$, and hence $\gcd(a,b,c) = 1$.  The proof will be similar for other values of $k$.
\end{proof}

\subsection{Upper Bound on $\lambda_k$ and approximation ratio} \label{upper}

\begin{theorem} \label{theorem-3}
We have $\lambda_k\leq \lambda$, with $\lambda$ given by (**). Consequently, the approximation ratio for the problem is not greater than $\frac{9}{8}$.
\end{theorem}
\begin{proof}
The first statement follows directly from Theorem~\ref{theorem-2}:  $\lambda_k\leq\lambda$ for any $\lambda$-labeling.
The approximation ratio is the ratio between the upper bound (UB), given by $\lambda$ from (**), and the lower bound (LB), given in Theorem~\ref{theorem-LB}. Note that for all the cases mentioned in Formula~\ref{formulae}, $\displaystyle{\lim_{p \to \infty}} \frac{UB}{LB} = \frac{9}{8}$.
\end{proof}

\section{Conclusion} \label{conclusion}
In this paper $\lambda$-$L(k, k-1, \ldots, 2, 1)$-labeling for square grid is proposed and the lower bound on $\lambda_k$, the $L(k, k-1, \ldots, 2, 1)$-labeling number,  is computed. A formula for a no-hole $\lambda$-$L(k, k-1, \ldots, 2, 1)$-labeling of square grid is given, implying  at most $\frac{9}{8}$ approximation ratio. The correctness proof of the proposed formula is given and it is also proved that the proposed formula gives a no-hole labeling.

\section*{Acknowledgement}\label{sec:ack}
We would like to thank Adam Paszkiewicz, Faculty of Mathematics and Computer Science, University of \L\'od\'z, \L\'od\'z, Poland for fruitful discussions on the lower bound estimate.

\section*{References}
\bibliography{SquareGrid}

\begin{thebibliography}{10}

\bibitem{hale1980frequency}
William~K Hale.
\newblock {Frequency assignment: Theory and applications}.
\newblock {\em Proceedings of the IEEE}, 68(12):1497--1514, 1980.

\bibitem{dubhashi2002channel}
Aniket Dubhashi, Madhusudana~VS Shashanka, Amrita Pati, R~Shashank, and Anil~M
  Shende.
\newblock {Channel assignment for wireless networks modelled as d-dimensional
  square grids}.
\newblock In {\em {International Workshop on Distributed Computing}}, pages
  130--141. Springer, 2002.

\bibitem{bertossi2003channel}
Alan~A. Bertossi, Cristina~M Pinotti, and Richard~B. Tan.
\newblock {Channel assignment with separation for interference avoidance in
  wireless networks}.
\newblock {\em IEEE Transactions on Parallel and Distributed Systems},
  14(3):222--235, 2003.

\bibitem{yeh1990}
R.K. Yeh.
\newblock {\em {Labeling graphs with a condition at distance two}}.
\newblock PhD thesis, Department of Mathematics, University of South Carolina,
  1990.

\bibitem{griggs1992labelling}
Jerrold~R Griggs and Roger~K Yeh.
\newblock {Labelling graphs with a condition at distance 2}.
\newblock {\em SIAM Journal on Discrete Mathematics}, 5(4):586--595, 1992.

\bibitem{chang19962}
Gerard~J Chang and David Kuo.
\newblock {The L(2,1)-labeling problem on graphs}.
\newblock {\em SIAM Journal on Discrete Mathematics}, 9(2):309--316, 1996.

\bibitem{bodlaender2000lambda}
Hans~L Bodlaender, Ton Kloks, Richard~B Tan, and Jan van Leeuwen.
\newblock {$\lambda$-coloring of graphs}.
\newblock In {\em {STACS 2000}}, pages 395--406. Springer, 2000.

\bibitem{chang2000d}
Gerard~J Chang, Wen-Tsai Ke, David Kuo, Daphne D-F Liu, and Roger~K Yeh.
\newblock {On L (d, 1)-labelings of graphs}.
\newblock {\em Discrete Mathematics}, 220(1):57--66, 2000.

\bibitem{shao20072}
Zhendong Shao and Roger~K Yeh.
\newblock {The L (2, 1)-labeling on planar graphs}.
\newblock {\em Applied mathematics letters}, 20(2):222--226, 2007.

\bibitem{chang2003distance}
Gerard~J Chang and Changhong Lu.
\newblock {Distance-two labelings of graphs}.
\newblock {\em European Journal of Combinatorics}, 24(1):53--58, 2003.

\bibitem{georges2003labeling}
John~P Georges and David~W Mauro.
\newblock {Labeling trees with a condition at distance two}.
\newblock {\em Discrete Mathematics}, 269(1):127--148, 2003.

\bibitem{calamoneri2004h}
Tiziana Calamoneri and Rossella Petreschi.
\newblock {L (h, 1)-labeling subclasses of planar graphs}.
\newblock {\em Journal of Parallel and Distributed Computing}, 64(3):414--426,
  2004.

\bibitem{fishburn2003no}
Peter~C Fishburn and Fred~S Roberts.
\newblock {No-hole L (2, 1)-colorings}.
\newblock {\em Discrete applied mathematics}, 130(3):513--519, 2003.

\bibitem{yeh2006survey}
Roger~K Yeh.
\newblock {A survey on labeling graphs with a condition at distance two}.
\newblock {\em Discrete Mathematics}, 306(12):1217--1231, 2006.

\bibitem{kuo20042}
David Kuo and Jing-Ho Yan.
\newblock {On L (2, 1)-labelings of Cartesian products of paths and cycles}.
\newblock {\em Discrete Mathematics}, 283(1):137--144, 2004.

\bibitem{calamoneri2009h}
Tiziana Calamoneri, Saverio Caminiti, Rossella Petreschi, and Stephan Olariu.
\newblock {On the L (h, k)-labeling of co-comparability graphs and circular-arc
  graphs}.
\newblock {\em Networks}, 53(1):27--34, 2009.

\bibitem{calamoneri2006optimal}
Tiziana Calamoneri.
\newblock {Optimal L (h, k)-labeling of regular grids}.
\newblock {\em Discrete Mathematics and Theoretical Computer Science}, 8, 2006.

\bibitem{georges1994relating}
John~P Georges, David~W Mauro, and Marshall~A Whittlesey.
\newblock {Relating path coverings to vertex labellings with a condition at
  distance two}.
\newblock {\em Discrete Mathematics}, 135(1):103--111, 1994.

\bibitem{shao20082}
Zhendong Shao and David Zhang.
\newblock {The L (2, 1)-labeling on Cartesian sum of graphs}.
\newblock {\em Applied Mathematics Letters}, 21(8):843--848, 2008.

\bibitem{nandi2015k}
Soumen Nandi, Sagnik Sen, Sasthi~C Ghosh, and Sandip Das.
\newblock {On L (k, k- 1,{\ldots}, 1) labeling of triangular lattice}.
\newblock {\em Electronic Notes in Discrete Mathematics}, 48:281--288, 2015.

\end{thebibliography}
\bibliographystyle{unsrt}
\end{document}